\documentclass[11pt, a4paper]{article}

\usepackage[T1]{fontenc}
\usepackage{charter,eulervm}
\usepackage{tikz}
\usetikzlibrary{patterns,snakes, arrows}

\usepackage{pgfplots}
\usepgfplotslibrary{fillbetween}

\usepackage{amssymb}
\usepackage{nicefrac}
\usepackage{amsmath}
\usepackage{amsthm}
\usepackage[round]{natbib}
\usepackage{setspace}
\usepackage[a4paper,margin=1.5in,footskip=0.4in]{geometry}

\usepackage{setspace}
\newcommand{\argmax}{\operatorname{argmax}}

\newtheorem{thm}{Theorem}
\newtheorem{prop}{Proposition}

\theoremstyle{definition}
\newtheorem{rem}{Remark}

\title{Limit Orders and Knightian Uncertainty\footnote{Examples similar to our motivating example have been studied by Dominique P\"aper and Peter Habiger in their respective master theses written under the supervision of Christoph Kuzmics. We are grateful to both as well as to Sophie Bade, Patrick Beissner, Simone Cerreia-Vioglio, Peter Klibanoff, Nenad Kos, Michael Mandler, Frank Riedel, Klaus Ritzberger, Jan Werner, and Michael Zierhut, as well as to seminar audiences at the University of Bielefeld, Bocconi University, Caltech, UC Irvine, University of Minnesota, Vienna University of Business and Economics, and participants of the VII Hurwicz workshop in Warsaw, the Conference of Mechanism and Institution Design 2020, and the Risk, Uncertainty, and Decision Conference 2020 for valuable comments and suggestions.}}
\author{Michael Greinecker\thanks{University of Graz, Department of Economics, Universit\"atsstrasse 15, 8010 Graz, Austria, phone: +43 316 380 3454, email: \href{mailto:michael.greinecker@uni-graz.at}{michael.greinecker@uni-graz.at}} \and Christoph Kuzmics\thanks{University of Graz, Department of Economics, Universit\"atsstrasse 15, 8010 Graz, Austria, phone: +43 316 380 7111, email: \href{mailto:christoph.kuzmics@uni-graz.at}{christoph.kuzmics@uni-graz.at}}}

\usepackage{hyperref}
\usepackage{xcolor}
\definecolor{dark-red}{rgb}{0.4,0.15,0.15}
\definecolor{dark-blue}{rgb}{0.15,0.15,0.4}
\definecolor{medium-blue}{rgb}{0,0,0.5}
\hypersetup{
    colorlinks, linkcolor={dark-red},
    citecolor={dark-red}, urlcolor={dark-red}
}

\begin{document}


\maketitle

\begin{abstract}
A range of empirical puzzles in finance has been explained as a consequence of traders being averse to ambiguity. Ambiguity averse traders can behave in financial portfolio problems in ways that cannot be rationalized as maximizing subjective expected utility. However, this paper shows that when traders have access to limit orders, all investment behavior of an ambiguity-averse decision-maker is observationally equivalent to the behavior of a subjective expected utility maximizer with the same risk preferences; ambiguity aversion has no additional explanatory power.

\bigskip


\end{abstract}

\newpage

\section{Introduction}

We show that any portfolio choice that cannot be explained by subjective expected utility maximization can also not be explained by ambiguity aversion when decision-makers have access to limit orders. This is in sharp contrast to a situation where decision-makers trade at given known prices, the usual setting in the literature. In our model, a decision-maker faces uncertainty over the joint distribution of the price of an asset at the point of purchase and its final value. They can trade via limit orders that trade contingent on prices. Our results show that every portfolio choice that cannot be explained as the behavior of a subjective expected utility maximizer for a given Bernoulli utility function for some probabilistic belief is state-wise strictly dominated by some limit order. This implies that any ambiguity averse decision maker whose preferences satisfy the von Neumann Morgenstern axioms over the set of constant acts and a monotonicity axiom can only choose a limit order that could also be chosen by a subjective expected utility maximizer with the same risk attitude.\footnote{Note the important qualification ``with the same risk attitude.'' Under certain assumptions, trading without limit orders in the presence of ambiguity can also be represented by an expected utility trader with, however, a substantially higher degree of risk aversion, see e.g., \citet{maenhout2004robust} and \cite{skiadas2003robust,skiadas2013smooth}.} Ambiguity aversion has no additional explanatory power.

For the standard Bayesian paradigm, all uncertainty can be quantified by a single probability distribution, and a rational decision-maker maximizes their expected utility with respect to this distribution. Under the Bayesian paradigm, both the probability distribution and the utility function to be maximized can be derived from personal preferences, and rational decision-makers are taken to be subjective expected utility maximizers.

There is a variety of empirically observed phenomena in finance that are hard to rationalize with the standard Bayesian paradigm. On the whole, individuals tend to be more conservative in their financial trading (they trade less) than they would if they were subjective expected utility maximizers with the risk preferences that they display in other contexts. Ambiguity aversion preference models in various forms are able to generate more conservative trading behavior at the intensive margin (e.g., the smooth model of \citet{klibanoff05}) or, for some models, even at the extensive margin (e.g., the maxmin expected utility model of \citet{gilboa89}). By admitting more conservative trading behavior, ambiguity aversion has been used to explain a variety of quantitative and qualitative empirical puzzles, such as the financial market participation puzzle, see e.g., \citet{10.2307/2951685}, the home-bias investment puzzle, see e.g., \citet{epstein2003two} and \citet{uppal2003model}, the excess volatility puzzle of \citet{shiller1981stock} and \citet{leroy1981present}, see e.g., \citet{epstein1994intertemporal} and \citet{illeditsch2011ambiguous}, and the equity premium puzzle of \citet{MEHRA1985145}, see e.g., \citet{ju2012ambiguity}, \citet*{doi:10.3982/QE708}, \citet{BRENNER2018503}, and \citet{IZHAKIAN2020105001}. In many cases, subjective expected utility maximization can only explain the size of effects if one assumes agents to be overly risk-averse. For example, the equity premium puzzle concerns the overly high implied risk-aversion of investors that is needed when one attributes the historical difference between the returns of equity and Treasury bonds to an equilibrium risk premium in a parametric general equilibrium model.

The explanations of these puzzles based on ambiguity aversion depend crucially on how investors can trade. We show that when decision-makers can trade with common financial instruments, observable differences between ambiguity averse decision-makers and subjective expected utility maximizers vanish. In all the papers mentioned above, decision-makers only trade at a given known price. If, in contrast, decision-makers can set price contingent orders -- limit orders -- before prices materialize, then for every behavior that can occur in standard models of ambiguity aversion, there exists a Bayesian probabilistic belief at which their choice maximizes subjective expected utility for the same Bernoulli utility function. Ambiguity aversion is not identifiable, and the ambiguity premium must be zero.\footnote{Furthermore, and importantly, the risk premium is unaffected by ambiguity aversion. This is in contrast to results on traders without access to limit orders in which a decision maker's ambiguity aversion can be understood ``as if'' they have a higher degree of risk aversion, but behave otherwise like subjective expected utility maximizers, see e.g., \citet{maenhout2004robust}, \citet{skiadas2003robust} and \citet{skiadas2013smooth}.}

A wide variety of models of ambiguity aversion has been used in finance, and we make our argument robust to the choice of the underlying decision-theoretic model. Our argument applies to all those models of ambiguity aversion formulated in an \citet{anscombe63} style framework, with both uncertainty and calculable risk, that satisfy a weak monotonicity requirement, and that apply expected utility theory to choice problems that only involve calculable risk.\footnote{The most popular such framework is the one introduced by \citet{fishburn1970utility}.} The weak monotonicity requirement is simply that no choice can be dominated by another single deterministic choice in each state of nature. The models of ambiguity aversion we allow for include, among others, the maxmin expected utility model of \citet{gilboa89}, the Choquet expected utility model of \citet{schmeidler89}, the smooth ambiguity model of \citet*{klibanoff05}, the variational and multiplier preference models of \citet*{maccheroni06} and \citet{hansen2001robust}, confidence function preferences of \citet{chateauneuf2009ambiguity}, uncertainty aversion preferences of \citet*{cerreia11}, and the incomplete preference model of \citet{bewley02}.\footnote{Our argument does not apply to some more behaviorally flavored models of ambiguity aversion such as those of \citet{seo09}, \citet{saito15}, and \citet{ke2017randomization} that violate the weak monotonicity condition. Such models have not been used in finance and lack normative appeal.}

To make our argument robust to the choice of the decision-theoretic model, we make use of a dual characterization of subjective expected utility maximization. In all the models listed above, choices that can be explained by the model but not by subjective expected utility maximization must be strictly dominated by randomized choices, but not by deterministic choices.\footnote{Technically, randomization convexifies the available choice set.} We can explain the logic of the last statement with a variant of one of the thought experiments of \citet{ellsberg61}: Consider two urns, each containing a hundred balls. One urn is unambiguous and is known to contain $49$ white balls and $51$ black balls. The other, ambiguous, urn is filled with $100$ balls that are each white or black, but the composition is not known. A decision-maker has to choose between the following bets: In bet one, the decision-maker wins a prize if a ball drawn from the ambiguous urn is black. In bet two, the decision-maker wins the same prize if a ball drawn from the ambiguous urn is white. Finally, in bet three, the decision-maker wins the same prize if a ball drawn from the unambiguous urn is white. A decision-maker choosing bet three must be ambiguity-averse, for their winning chance is only $\nicefrac{49}{100}$, while for any probabilistic belief about the composition of the ambiguous urn, either bet one or two (or both) must have a winning chance of at least $\nicefrac{1}{2}$.

Note that none of the three bets dominates another in pairwise comparisons for every composition of the ambiguous urn. But as \citet{raiffa61} pointed out, a lottery that chooses bets one and two with probability $\nicefrac{1}{2}$ each has a higher winning chance (of $\nicefrac{1}{2}$) than bet three independently of the composition of the ambiguous urn. Only those two bets that can be chosen by a subjective expected utility maximizer for some probabilistic belief over the urn composition are not dominated by a randomized bet.

As \citet{kuzmics2017abraham} pointed out, this is true in general. A lemma of \citet{pearce84} states that in a finite two-player game in normal form, a strategy is not dominated by any mixed strategy if and only if it is a best reply to a mixed strategy of the opponent. Related results can be traced back all the way to the complete class theorem of \citet{wald47annals} in statistical decision theory. The lemma of Pearce, when translated to a decision problem in the setting of \citet{anscombe63} with finitely many states, says that every decision that does not maximize subjective expected utility with respect to some probabilistic belief over the states of nature must be strictly dominated by a randomized choice. The models of ambiguity aversion listed above do not allow for choices that are dominated by deterministic, nonrandomized choices. Therefore, the only choices that can be explained by ambiguity aversion but not by subjective expected utility maximization are those that are strictly dominated by randomized choices but not by nonrandomized choices. We show that in our setting, there are no such choices. Thus, even if the decision maker does not randomize, their choices must be consistent with subjective expected utility for their given risk-preferences. We first illustrate this claim in a simple example that abstracts from various features of our general model, such as general risk aversion or informative prices, but allows for a very simple argument.

For the general model, we need to adapt our arguments slightly since we are working in an infinite-dimensional setting. Spaces of limit orders are infinite-dimensional function spaces, and the state space is allowed to be infinite. Nevertheless, we show in Proposition \ref{wald}, using a generalization of Pearce's lemma due to \citet*{battigalli16n}, that a limit order is not dominated by a randomly chosen limit order if and only if it maximizes subjective expected utility with respect to some probabilistic belief. While Proposition \ref{wald} corresponds to a universal decision theoretic principle, Proposition \ref{det}, our main technical result, makes crucial use of our finance setting. It states that a limit order that is strictly dominated by a randomly chosen limit order, a mixed limit order in our language, must already be strictly dominated by a single (deterministic) limit order. Therefore, as we state in the main Theorem \ref{newthm}, everything that can be explained by ambiguity aversion could be explained by subjective expected utility maximization with the same Bernoulli utility function as well.

For many models of ambiguity aversion, an even stronger conclusion holds than the one given in  Theorem \ref{newthm}. Proposition \ref{beliefs} shows that behavior that can be explained by ambiguity aversion based on a model with a set of probability distributions to model beliefs can be explained as maximizing expected utility with respect to a probabilistic belief in the closed convex hull of the underlying set of probability distributions. Our argument does, therefore, not rely on decision makers having extreme probabilistic beliefs.  Proposition \ref{simp} shows that the limit orders in our arguments can be implemented by realistic market portfolios.
\bigskip

Here is the structure of the rest of this paper: Section \ref{sec:example} provides a simple example that shows the dramatic difference limit orders can make. Section \ref{sec:model} presents the environment we use and our main results. Section \ref{sec:disc} discusses the scope of our results, the modeling choices we make, and how one can generalize some results. Section \ref{sec:proofs} collects all proofs, except for the proof of Theorem \ref{newthm} which is a direct consequence of Propositions \ref{wald} and \ref{det}.

\section{Motivating Example}\label{sec:example}

We illustrate our central point in terms of a very simple portfolio problem. There is one asset whose future value can be either $-1$ or $1$. A risk-neutral decision-maker is faced with a price of $p$ and asked to buy or (short) sell at most one unit of the asset. The amount the decision-maker can buy must, therefore, lie in the interval $[-1,1]$. If the decision-maker maximizes subjective expected utility with respect to some probabilistic belief $\beta$ on the value being $1$, buying one unit is uniquely optimal if $p<\beta-(1-\beta)$, selling one unit is uniquely optimal if $p>\beta-(1-\beta)$, and everything is optimal, including not trading at all, if $p=\beta-(1-\beta)$.

Suppose now that the decision-maker behaves according to the maxmin expected utility model of \citet{gilboa89} and that two beliefs $\beta_l<\beta_h$ form the extreme points of their set of beliefs. Now buying one unit is optimal if $p\leq\beta_l-(1-\beta_l)$, selling one unit is optimal if $p\geq \beta_h-(1-\beta_h)$, and not trading is uniquely optimal when $p$ lies in the interval $\big(\beta_l-(1-\beta_l),\beta_h-(1-\beta_h)\big)$. Ambiguity aversion allows a decision-maker completely to completely refrain from trading for a range of prices. This is the core of the argument of \citet{10.2307/2951685}, who explain observed stock market abstention this way. Not all models of ambiguity aversion give rise to such abstention. Consider a decision-maker who behaves according to the smooth ambiguity aversion model of \citet{klibanoff05}. Such a decision-maker has a prior over the set of probabilistic beliefs, applies an increasing transformation to expected utilities under each probabilistic belief, and then takes the average with respect to the prior. For example, following Example 1 of \citet{klibanoff05}, we can take a risk-neutral decision-maker whose prior assigns probability $1/2$ each to the two beliefs $\beta_l=1/4$ and $\beta_h=3/4$, and whose increasing transformation is given by a function $\phi(x) = (1-e^{-\alpha x})/(1-e^{-\alpha})$ with $\alpha>0$. In the terminology of  \citet{klibanoff05}, the decision-maker has constant ambiguity aversion with $\alpha$ the coefficient of ambiguity aversion. The payoff of the decision-maker when buying $x$ units of the asset at price $p$ is then
\[\nicefrac12 ~\phi\Big(x\big(\nicefrac14~(-1-p) + \nicefrac34~ (1-p) \big)\Big) + \nicefrac12 ~\phi\Big(x\big(\nicefrac34~(-1-p) + \nicefrac14~(1-p)\big)\Big).\]
The demand for the asset can be explicitly calculated and is given by a demand function $d:\mathbb{R}\to [-1,1]$ such that
\[d(p) = \left\{ \begin{array}{ll} 1 & \mbox{ if } p \le -\frac12, \\ \min\left\{\frac{\ln(1/2-p)-\ln(1/2+p)}{\alpha},1\right\} & \mbox{ if } p \in (-\frac12,0], \\ \max\left\{\frac{\ln(1/2-p)-\ln(1/2+p)}{\alpha},-1\right\} & \mbox{ if } p \in (0,\frac12), and \\ -1 & \mbox{ if } p \ge \frac12. \end{array}\right.\]
The decision-maker will generically trade a nonzero amount, but the volume traded will be much smaller than the (here neutral) risk preferences can explain under subjective expected utility. Such a reduction in the demand of the asset traded might be able to explain the equity premium puzzle as a consequence of ambiguity aversion.

So far, we have assumed that the decision-maker knows the price of the asset when trading and only one price becomes relevant.
We now consider a situation in which the decision-maker trades via (generalized) limit orders before prices realize. Formally, limit orders are measurable functions from $\mathbb{R}$ to the interval $[-1,1]$. For a limit order $l$ and a price $p$, we interpret $l(p)$ as the amount of the asset $l$ buys at $p$. We do not differentiate between two limit orders that agree almost surely. A decision-maker trading with limit orders must have a model of how prices relate to final values. For the sake of this example, and only for the sake of this example, we follow \citet{10.2307/2951685} and assume that prices do not tell us anything about the future value of the asset. Prices are distributed according to a distribution $F$ that has full support and no mass points. Finding the optimal choice for a subjective expected utility maximizer is straightforward. Their chosen limit order must maximize the conditional expected utility given the price for almost every price. So they must choose a limit order that buys one unit if $p<\beta-(1-\beta)$ and sells one unit if $p>\beta-(1-\beta)$, for almost every $p$. The limit order simply carries out what the decision-maker would have done if they knew the price of the asset.

The situation for an ambiguity averse decision maker is completely different, both under maxmin expected utility and the smooth model. Neither wants to choose a limit order that mimics their demand. Actually, only limit orders that a subjective expected utility maximizer could also choose can be chosen by such an ambiguity averse decision maker. A subjective expected utility maximizer can only use a cut-off limit order $l_v$ with cut-off $v$ such that $l_v(p)=1$ for $p<v$ and $l_v(p)=-1$ for $p\geq v$. Indeed, only such cut-off limit orders with cut-off $v \in [-1,1]$ are undominated. To see this, let $l:\mathbb{R}\to [-1,1]$ be a limit order such that $-1<l(p)<1$ for a set of prices $p$ with positive measure or such that $l$ is not almost everywhere equal to a nondecreasing function. We can assume without loss of generality that $l(p)=1$ for $p\leq-1$ and $l(p)=-1$ for $p\geq1$; limit orders violating this condition are trivially dominated. Now, there must then exist some value $v$ in the interval $[-1,1]$ such that

\[\int_{-1}^v \left(1-l(p)\right)~\mathrm d F=\int_{v}^1 \left(1+l(p)\right)~\mathrm d F.\]
Indeed, if we let $L(v)$ and $R(v)$ denote the left and the right-hand side, respectively, then $\lim_{v\downarrow -1}L(v)=0$, $\lim_{v\uparrow 1}L(v)>0$, $\lim_{v\downarrow -1}R(v)>0$, and $\lim_{v\uparrow 1} R(v)=0$. Moreover, $L$ and $R$ are continuous functions since $F$ has no mass points. The existence of a $v$ such that $L(v)=R(v)$ follows, therefore, from the intermediate value theorem; see the picture below.\bigskip

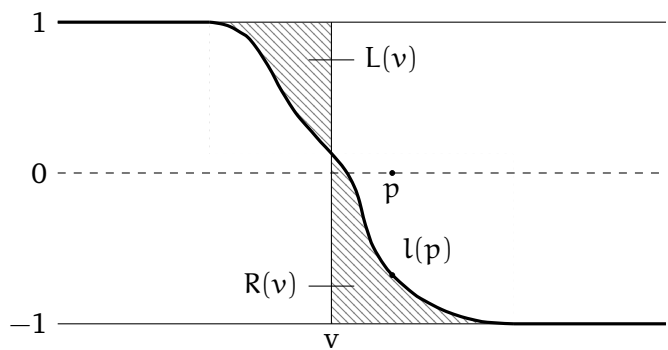
\begin{tikzpicture}
\draw[pattern=north west lines, pattern color=gray	, draw=white] (-2,2) rectangle (-0.4,0.26);
\draw[pattern=north west lines, pattern color=gray	, draw=white] (2,-2) rectangle (-0.4,0.26);

\draw  plot[smooth, tension=.7] coordinates {(-2,2) (-1.86,1.98) (-1.73,1.94) (-1.59,1.86) (-1.5,1.8) (-1.4,1.68) (-1.3,1.53) (-1.18,1.33) (-1.1,1.17) (-1,1) (-0.9,0.84) (-0.81,0.72) (-0.66,0.55) (-0.54,0.41) (-0.4,0.26) (-0.22,0.04) (-0.1,-0.18) (-0.03,-0.38) (0.02,-0.59) (0.08,-0.79) (0.15,-0.98) (0.3,-1.23) (0.44,-1.39) (0.68,-1.6) (0.82,-1.69) (1.05,-1.81) (1.26,-1.89) (1.44,-1.94) (1.65,-1.98) (2,-2)};

\filldraw[white, fill=white]  plot[smooth, tension=.7] coordinates {(-2,2) (-1.86,1.98) (-1.73,1.94) (-1.59,1.86) (-1.5,1.8) (-1.4,1.68) (-1.3,1.53) (-1.18,1.33) (-1.1,1.17) (-1,1) (-0.9,0.84) (-0.81,0.72) (-0.66,0.55) (-0.54,0.41) (-0.4,0.26) (-0.5,0.26)(-0.7,0.26)(-2,0.26)};

\filldraw[white, draw=none]  plot[smooth, tension=.7] coordinates {(-0.4,0.26) (-0.22,0.04) (-0.1,-0.18) (-0.03,-0.38) (0.02,-0.59) (0.08,-0.79) (0.15,-0.98) (0.3,-1.23) (0.44,-1.39) (0.68,-1.6) (0.82,-1.69) (1.05,-1.81) (1.26,-1.89) (1.44,-1.94) (1.65,-1.98) (2,-2) (2,-1.9) (2,-1.8) (2,-1.7) (2,0.26)};

\draw[-] (-4,2)node[left]{$1$} -- (4,2) node{};
\draw[-, very thick] (-4,2)node{} -- (-2,2) node{};
\draw[-, dashed] (-4,0)node[left]{$0$} -- (4,0) node{};
\draw[-] (-4,-2)node[left]{$-1$} -- (4,-2) node{};
\draw[-, very thick] (2,-2)node{} -- (4,-2) node{};
\draw[-, thin] (-0.7,1.5)node{} -- (-0.1,1.5) node[right]{$L(v)$};
\draw[-, thin] (-0.1,-1.5) -- (-0.7,-1.5) node[left]{$R(v)$};

\draw[fill] (0.4,-1.354) circle (.17ex) node[above right]{$l(p)$};
\draw[fill] (0.4,0) circle (.17ex) node[below]{$p$};

\draw[-] (-0.4,2)node{} -- (-0.4,-2) node[below]{v};
\draw[very thick]  plot[smooth, tension=.7] coordinates {(-2,2) (-1.86,1.98) (-1.73,1.94) (-1.59,1.86) (-1.5,1.8) (-1.4,1.68) (-1.3,1.53) (-1.18,1.33) (-1.1,1.17) (-1,1) (-0.9,0.84) (-0.81,0.72) (-0.66,0.55) (-0.54,0.41) (-0.4,0.26) (-0.22,0.04) (-0.1,-0.18) (-0.03,-0.38) (0.02,-0.59) (0.08,-0.79) (0.15,-0.98) (0.3,-1.23) (0.44,-1.39) (0.68,-1.6) (0.82,-1.69) (1.05,-1.81) (1.26,-1.89) (1.44,-1.94) (1.65,-1.98) (2,-2)};
\end{tikzpicture}\bigskip

We claim that the cut-off limit order $l_v$ strictly dominates the limit order $l$; it gives a higher expected return in both states. Actually, it gives the same additional positive expected return independently of the final value. To see this, note that $l_v$ buys an additional expected amount of $L(v)$ and sells the same additional expected amount. But this additional expected amount is bought when prices are relatively low (at most $v$) and sold when prices are relatively high. This leads mechanically to an additional positive expected return that does not depend on the final value. So, the only limit orders that are not strictly dominated are those that a subjective expected utility maximizer that has the same risk-preferences might also choose. Consequently, a decision maker that behaves according to the maxmin expected utility model, the smooth ambiguity model, or any similar model will be observationally indistinguishable from a subjective expected utility maximizer---even with the same risk preferences.\bigskip

It is tempting to explain the difference limit orders make for ambiguity averse decision makers by peculiarities of the example: Prices are uninformative; there is only one asset; the decision-maker faces fundamental uncertainty about final values, but somehow not about prices arising in-between; the decision-maker is risk-neutral. None of that matters. As we show in the next section in a setting in which the decision-maker faces ambiguity about the joint distribution of prices and final values, every choice of a limit order that cannot be rationalized as maximizing expected utility with respect to some probabilistic belief must be strictly dominated by a single, deterministic limit order. None of the standard theories of ambiguity aversion allows for choices that are dominated by deterministic choices, so they are unable to explain anything that cannot already be explained by subjective expected utility maximization. The price of the increased generality is that our proof is not constructive. For notational ease we restrict our model to a single asset. None of our arguments depend on this restriction. Of course, real-life portfolio decisions require joint optimization over all available assets, even in settings of pure risk.\bigskip

An alternative interpretation of the result by \citet{10.2307/2951685} is that ambiguity averse decision-makers use the maxmin criterion not ex-ante before the price of the asset realizes, but at an interim stage at which prices are already known but final values are not. Such a decision-maker considers a range of distributions over final values possible conditional at each price and assumes the worst such distribution for each price. Our model in Section \ref{sec:model} cannot fully accommodate such a situation; our restriction on the joint distribution of prices and final values admitting well-behaved densities limits the extent to which the set of priors can vary after conditioning on the price. However, even in a model in which we allow the decision-maker to behave like a maxmin expected utility maximizer conditional on the price, all behavior can still be rationalized as the behavior of a subjective expected utility maximizer, and the rationalizing  belief can be constructed manually. The machinery based on our technical assumptions is not needed. We show this in an appendix at the end. In the model that we utilize there, we fix the marginal distribution of prices as in the motivating example but allow among the states all deterministic functions relating prices and final values. An ambiguity averse decision maker could refrain from trading at all prices between $-1$ and $1$. Importantly, however, a subjective expected utility maximizer can do so too. The belief that would support this lack of trading is the belief that the conditional expected final value at each price equals the price -- the decision-maker believes markets to be informationally efficient.

\section{Environment and Main Result} \label{sec:model}

The decision-maker faces uncertainty over which probabilistic model best describes the relationship between the price of an asset and its final value. We think of these as possible objectively correct models; all residual uncertainty conditional on the true model is objective, quantifiable risk. We follow \citet{10.2307/1882087} here, who considered only those matters fundamentally uncertain for which ``there is no scientific basis on which to form any calculable probability whatever.''  Uncertainty within scientific models is taken to be objective uncertainty.

There is a compact metrizable space $Y$ of models and for each model a joint density over prices and final values of the single asset with respect to some product measure $\pi\otimes\xi$ obtained from $\sigma$-finite Borel measures $\pi$ and $\xi$ on $\mathbb{R}$.  The family of densities can be represented by a single nonnegative measurable function $h:\mathbb{R}\times\mathbb{R}\times Y\to\mathbb{R}$ continuous in $Y$ such that
\[\int\int h(p,x,y)~\mathrm d \pi(p)~\mathrm d\xi(x)=1\] for all $y\in Y$.

The decision-maker's ultimate payoff from any net-gains from investing is given by an increasing continuous Bernoulli utility function $u:\mathbb{R}\to\mathbb{R}$.

\begin{rem}
We assume that the Bernoulli utility function is defined on the whole real line. This rules out certain Bernoulli utility functions such as logarithmic ones, and requires the Bernoulli utility function to be unbounded if it should be increasing and weakly concave. One could easily change the framework to have Bernoulli utility functions on domains that are bounded below. Later, the decision-maker is allowed to short-sell and, therefore, able to make large losses. If we want to have Bernoulli utility functions on a restricted domain, we would simply need an assumption that guarantees that final wealth levels stay in the domain of the Bernoulli utility function for all losses that can occur.
\end{rem}

The net amount the decision-maker is allowed to invest is restricted to lie in an interval $[b,t]$ with $b<0<t$. The bounds represent short-selling and debt constraints, respectively.
To guarantee that expected utilities are defined and finite, we assume that there is a $\pi\otimes\xi$-integrable function $d:\mathbb{R}\times\mathbb{R}\to\mathbb{R}$
satisfying
\[|u(tx-tp)|h(p,x,y)+|u(bx-bp)|h(p,x,y)\leq d(p,x)\]
for all $(p,x)\in\mathbb{R}^2$ and all $y\in Y$.

\begin{rem}This is a joint assumption on the Bernoulli utility function and the possible joint distributions over prices and payoffs. For bounded $u$, it holds automatically. If $u$ is increasing and weakly concave, then the less concave the Bernoulli utility function is, the more restrictive is the assumption on the joint distribution over prices and payoffs. For the extreme case of a risk-neutral agent, it amounts to a uniform integrability condition on net-gains. Even then, the condition is fairly weak and mainly requires that the tails of all distributions vanish sufficiently fast in a uniform way. For example, $h$ could be the density of a bivariate normal distribution or bivariate t-distribution with degrees of freedom $2+\epsilon$ (with $\epsilon>0$ arbitrarily small), and $Y$ a compact set of pairs of means and invertible covariance matrices.\footnote{These examples include, therefore, bivariate distributions with existing means and variances. Such distributions need not have a kurtosis, and, thus, include bivariate distributions with heavy tails.} A possible choice of the bounding function $d$ for all these cases would be a scaled-up joint density of two independent t-distributions with degrees of freedom strictly between $1$ and $1+\epsilon$.
\end{rem}

We assume that the decision-maker acts before prices are known and chooses a limit order. A \emph{limit order} is a measurable function from $\mathbb{R}$ into the set $A=[b,t]$. More precisely, limit orders are taken to be equivalence classes of measurable functions from $\mathbb{R}$ to $A$ with two such functions being equivalent if they agree outside a set of $\pi$-measure zero. We denote the set of limit orders by $L$ and endow it with the topology of convergence in measure\footnote{Essentially, we use the topology of convergence in probability for any probability measure mutually absolutely continuous with respect to $\pi$. The resulting topology does not depend on the specific choice of the probability measure and is separable and completely metrizable.} and its corresponding Borel $\sigma$-algebra. We embed $L$ in the space of \emph{mixed limit orders} $\Delta(L)$, the space of Borel probability measures on $L$, by identifying the deterministic limit order $l$ with the Dirac point mass $\delta_l$ concentrated on $l$. A \emph{deterministic limit} order is a mixed limit order of the form $\delta_l$.

\begin{rem}
Our arguments make use of mixed limit orders, but, as we see below, one can interpret them as purely ancillary mathematical constructs used in proofs. The main point is that whether a limit order is undominated or not does not depend on the availability of mixed limit orders. Our main result, Theorem \ref{newthm}, does not refer to any mixed limit orders.
\end{rem}

The expected payoff from $\mu\in\Delta(L)$ if the true model is $y$ is
\[V(\mu,y)=\int\int\int u\big(l(p)(x-p)\big) h(x,p,y)~\mathrm d \pi(p)~\mathrm d \xi(x)~\mathrm d\mu(l).\]

We say that $\mu'\in\Delta(L)$ \emph{strictly dominates} $\mu\in\Delta(L)$ if $V(\mu',y)>V(\mu,y)$ for all $y\in Y$. A mixed limit order that is not strictly dominated by another mixed limit order is \emph{undominated}.  A mixed limit order that is not strictly dominated by a nonrandomized limit order is \emph{deterministically undominated}. The following proposition gives a version of the familiar characterization of undominated choices from statistical decision theory and game theory.

\begin{prop}\label{wald}A mixed limit order $\mu$ is undominated if and only if there exists a probability distribution $\beta\in\Delta(Y)$ such that
$\mu$ maximizes $\int V(\cdot,y)~\mathrm d\beta(y)$.
\end{prop}

\noindent We are now ready for our central technical result.\footnote{The idea to base the argument on the concavity of $u$ was anonymously suggested to us.}

\begin{prop}\label{det}Assume that $\pi$ is atomless or $u$ is weakly concave. Then a mixed limit order is undominated if and only if it is deterministically undominated.
\end{prop}

Here is the basic idea: Suppose a mixed limit order $\mu$ is dominated by a mixed limit order $\mu'$. Instead of choosing a deterministic limit order at random, as $\mu'$ does, one could imagine randomizing conditionally on the price by ``behavioral limit orders.'' That this makes no difference follows from results in \citet{balder1981mathematical} or \citet{MR658713}. They are versions of the result of \citet{MR0054924} that mixed and behavioral strategies induce the same distributions over plays in extensive form games of perfect recall, or, even closer, the result of \citet{MR42657} on the equivalence of these two forms of randomizing in statistical decision theory. There is a natural topology on ``behavioral limit orders'' under which payoffs are continuous and in which deterministic limit orders are dense if $\pi$ is atomless. This corresponds to the denseness of controls in \citet{MR0372708} or the denseness of pure strategies in \citet{MR812820}. By this denseness, one can find a deterministic limit order that still dominates $\mu$. If $\pi$ is not atomless but $u$ concave, one could replace a behavioral limit order by the limit order that trades the expected amount of the behavioral limit order at each price. For risk averse decision-makers, this is an improvement for every distribution on prices and final values.
 \bigskip

In the remainder of this section we express the key consequence of Propositions \ref{wald} and \ref{det} in purely decision theoretic terms. The state space is the set $Y$ and the set of outcomes, the possible financial gains, is the set $\mathbb{R}$. Each limit order $l$ induces a unique Anscombe-Aumann act $f_l:Y \to \Delta(\mathbb{R})$ given by
\[\int g(r)~\mathrm df_l(y)(r)=\int\int g\big(l(p)(x-p)\big) h(p,x,y)~\mathrm d \pi(p)~\mathrm d \xi(x), \]
for each bounded measurable function $g:\mathbb{R}\to\mathbb{R}$. That this defines $f_l(y)$ follows from \citet[Theorem 4.5.2.]{MR1932358}.

Let $\mathcal{L}$ be the set of all acts that are induced by some $l \in L$. We call a set of (measurable) acts $\mathcal{F}$ \emph{rich} if it includes $\mathcal{L}$ as well as all constant acts with values of the form $f_l(y)$ for some $l \in L$ and $y \in Y$.

Let $\succeq$ denote a (not necessarily complete or transitive) preference relation on $\mathcal{F}$. We say that $\succeq$ is \emph{compatible} with the Bernoulli utility function $u$ if the restriction of $\succeq$ to the set of constant acts in $\mathcal{F}$ can be represented by the expectation of $u$ with respect to (the values of) these constant acts. That $\succeq$ is compatible with a Bernoulli utility function essentially amounts to $\mathcal{F}$ being large enough to include the convex hull of its constant acts, and for the von Neumann Morgenstern axioms to hold on the set of constant acts in $\mathcal{F}$.\footnote{We also need an additional continuity assumption for $u$ to be continuous and integrable for all relevant distributions. Since we allow for unbounded Bernoulli utility functions, the appropriate continuity notion must relate to the allowed probability distributions; see \citet{MR3212197} for an elegant approach.}

For any act $f \in \mathcal{F}$ and any state $y \in Y$ let $f^{y}$ denote the constant act that satisfies $f^{y}(y')=f(y)$ for all $y' \in Y$. We say that the preference relation $\succeq$ (with strict part $\succ$) is \emph{monotone} if $f \succ g$ whenever $f^{y} \succ g^{y}$ for every state $y \in Y$.

With this in place we can state our main result.

\begin{thm}\label{newthm} Assume that $\pi$ is atomless or $u$ is weakly concave.
Let $\succeq$ be a monotone (not necessarily complete or transitive) preference relation on a rich set of acts $\mathcal{F}$ that is compatible with the Bernoulli utility function $u$, and let $l$ be a $\succeq$-maximal element in the set of limit orders. Then there exists a probability distribution $\beta\in\Delta(Y)$ such that $l$ maximizes $\int \int u(m)~\mathrm d f_l(y)(m)~\mathrm d\beta(y)$.
\end{thm}
\begin{proof}Since $\succeq$ is monotone and compatible with $u$, the limit order $l$ must be deterministically undominated. By Proposition \ref{det}, $l$ is undominated. By Proposition \ref{wald}, there exists $\beta\in\Delta(Y)$ such that
$l$ maximizes
 \[\begin{split}
\int V(l,y)~\mathrm d\beta(y)&=\int\int\int u\big(l(p)(x-p)\big) h(x,p,y)~\mathrm d \pi(p)~\mathrm d \xi(x)~\mathrm d\beta(y)\\
&=\int \int u(m)~\mathrm d f_l(y)(m)~\mathrm d\beta(y).
\end{split}\]
\end{proof}

Theorem \ref{newthm} is our central result. It shows that every choice of a limit order of an ambiguity-averse decision-maker might as well be explained as the decision of a subjective expected utility maximizer with the same Bernoulli utility function for a suitably chosen probabilistic belief.

\section{Discussion} \label{sec:disc}

Many models of ambiguity aversion represent uncertainty by a set $\Pi$ of probability distributions over the states of nature. For such models, one can ask whether our rationalizing probabilistic belief may need to be more extreme than every member of $\Pi$. It does not. For sets of probability distributions to be interpretable as a collection of reasonable beliefs, it should be the case that a decision-maker will prefer one act over another whenever the former gives a higher expected utility with respect to every probability distribution in $\Pi$. Such a decision-maker's preferences can be interpreted as an extension of the unanimity ordering of \citet{bewley02}. It, therefore, suffices to show that rationalizing probabilistic beliefs can be chosen to be no more extreme than every member of the set of probability distributions when preferences respect the unanimity ordering. We do so in the next proposition, which shows that behavior can then be rationalized as maximizing expected utility with respect to some belief in the closed convex hull of $\Pi$.

\begin{prop}\label{beliefs}Assume that $\Pi\subseteq\Delta(Y)$ is nonempty and weak*-closed. Let $l$ be a deterministic limit order and assume that there is no deterministic limit order $l'$ such that $\int V(l',y)~\mathrm d\beta(y)>\int V(l,y)~\mathrm d\beta(y)$ for all $\beta\in\Pi$. Then there exists a probability measure $\beta$ on $Y$ in the weak*-closed convex hull of $\Pi$ such that $l$ maximizes $\int V(\cdot,y)~\mathrm d\beta(y)$.
\end{prop}
If the set of probability distributions on $Y$ is already closed and convex, as it usually is in the corresponding representation results, the rationalizing belief can be chosen out of the set of probability distributions itself.

At the heart of Proposition \ref{beliefs} is a generalization of Proposition \ref{wald} that holds outside our finance setting and that might be of independent interest: If $\Pi$ is nonempty and compact and a mixed act $\mu$ is maximal under the unanimity ordering, then there exists a probabilistic belief $\beta$ on $Y$ in the closed convex hull of $\Pi$ such that $\mu$ maximizes expected utility with respect to $\beta$. To prove this, one simply applies a suitable version of the lemma of \citet{pearce84} to a model in which the set of states of nature is replaced by $\Pi$. If one interprets elements of $\Pi$ as states this way, being undominated is equivalent to being maximal in the unanimity order. So for a maximal mixed act, there exists a probabilistic belief over $\Pi$ that rationalizes the choice of $\mu$. This probabilistic belief over members of $\Pi$ averages out to the desired probabilistic belief $\beta$ on $Y$ in the closed convex hull of $\Pi$. To finish the proof of Proposition \ref{beliefs}, one combines this generalization of Proposition \ref{wald} with arguments similar to those in the proof of Proposition \ref{det}. \bigskip

Our model has an implicit time-line: first the decision-maker chooses a limit order, then the price realizes, trade is executed, before, finally, the payoff realizes. Yet, it is a static model in which the decision-maker acts exactly once and before receiving any information. To model the same economic environment without limit orders one has to make the implicit dynamic explicit and also choose some dynamic version of the corresponding theory of ambiguity aversion. Most dynamic models of ambiguity aversion are consequentialist extensions of static models, see e.g., \cite{siniscalchi2011dynamic} and the references therein.\footnote{\cite{siniscalchi2011dynamic} also shows that dynamic models of ambiguity aversion necessarily have to violate either dynamic consistency, or consequentialism, or restrict the domain of preferences or allowed decision problems.}

In consequentialist models, ambiguity aversion gives rise to dynamic inconsistencies. Such a dynamically inconsistent decision-maker has a demand for commitment, with limit orders serving as commitment devices. Without commitment devices, they would engage in ex-ante strictly dominated behavior. These decision-makers would potentially renege on a chosen limit order after receiving new information, information that such decision-makers would, therefore, try to avoid. A negative value of information is natural when decision-makers have a demand for commitment, see, for example, \cite{siniscalchi_2009}. If a sophisticated decision-maker does not use a limit order, when available, they must be dynamically consistent to begin with and, therefore, behave like a subjective expected utility maximizer. These insights go beyond our model: the availability of financial instruments that can serve as commitment devices, such as limit orders in our case, can make a dramatic difference in finance models with ambiguity averse traders, even if the availability of the same financial instruments makes no difference to subjective expected utility traders.

\cite{hanany2007updating,MR2566600} provide a fundamentally different dynamic extension of static preference models of ambiguity aversion, one that maintains dynamic consistency at the price of consequentialism. Our main result still applies to this setting, of course. However, it has further implications for the counterfactual without limit orders. By dynamic consistency, such a decision-maker would behave the same way in the dynamic model with and without limit orders. The simple fact that prices are ex-ante uncertain (and given our dominance result) alone implies that a dynamically consistent decision-maker would implement their ex-ante optimal plan also after learning the price. Such a decision-maker would behave like a subject expected utility maximizer to begin with.

Either interpretation points to mechanisms that limit what ambiguity aversion can explain in finance models. Under the first interpretation one should consider the possibility that real traders have access to financial instruments that could serve as commitment devices, and this commitment matters for ambiguity averse traders. Under the second interpretation, the behavior of an ambiguity averse decision-maker is affected by how their problem is embedded in a larger context. For example, if a decision-maker faces an investment opportunity at a single given price without a prior stage, the argument of \cite{10.2307/2951685} applies, yet this decision-maker would behave like a subjective expected utility maximizer when the price that they are facing is chosen at some ex-ante level as it is in our model. \bigskip

One might worry that the dominating deterministic limit orders shown to exist in Proposition \ref{det} might be complicated measurable functions that do not correspond to anything that could be implemented in real financial markets.  This is not the case. Market participants have access to so-called \emph{stop-loss limit orders} that buy a fixed (possibly negative) amount for all prices within some interval of prices and nothing outside this interval. Positive linear combinations of such stop-loss limit orders represent, therefore, portfolios that are economically feasible. Each such portfolio corresponds to a linear combination of indicator functions of intervals. Negative weights represent stop-loss limit orders that (short) sell. We, therefore, define a \emph{simple limit order} to be a linear combination of indicator functions of intervals. These are essentially step-functions, and standard arguments for approximating general measurable functions by step-functions guarantee that we can find for each dominating deterministic limit order a dominating simple limit order.

\begin{prop}\label{simp}A mixed limit order is deterministically undominated if and only if it is not dominated by a simple limit order.
\end{prop}

Though we have shown that every undominated limit order must maximize expected utility for some belief, we have not shown that the maximization problem of a subjective expected utility maximizer has a solution. It does. Though the space of limit orders is not compact, the space of ``behavioral limit orders'' we have used is a compact convex set in a suitably chosen locally convex Hausdorff topological vector space, and, for a given belief, the function that maps a behavioral limit order to its expected utility is linear and continuous (compare the proof of Proposition \ref{wald}). Such a continuous linear function on a nonempty compact convex set must take its maximum on an extreme point, on a behavioral limit order that cannot be written as a proper mixture of two other behavioral limit orders. By \citet[Theorem 6.1]{balder1981mathematical} or the results in \citet[Section I]{MR658713}, such an extreme point must be a nonrandomized limit order.

Finally, the formalism we use deserves some discussion. For the proof of Lemma \ref{wald}, we need expected payoffs to be jointly continuous in mixed limit orders and beliefs over $Y$. Without some restriction on the dependence between the distributions of prices and payoffs, this is generally not possible. The approach we have taken is inspired by the existence results for Bayesian games of \citet{MR812820} and, in particular, \citet{MR942618}, whose framework and results we rely on. That the joint distribution of prices and payoffs is absolutely continuous with respect to a product measure corresponds to a diffuseness condition on types in Bayesian games. \citet{MR2888834} discusses discontinuities that can arise without such an assumption. Given the mathematical machinery, it is straightforward to modify the setting so that Bernoulli utility functions are defined only for positive wealth levels, so that state-dependent payoffs are allowed for, or so that the decision-maker can choose more than one asset.

For the sake of simplicity, we have implicitly assumed that the amount that the decision-maker is allowed to buy or sell of the asset does not depend on the price. If we interpret these bounds as debt and short-selling constraints, we should allow for dependence on the price. This could easily be done in our framework. Take now $\mathcal{A}:\mathbb{R}\to 2^\mathbb{R}$ to be a measurable correspondence with nonempty and compact values, and interpret $\mathcal{A}(p)$ as the set of possible net-amounts of the asset that can be bought at price $p$. By \citet[Corollary 1.2]{MR534114}, there exists a compact metrizable space $A$ and a measurable function $c:A\times\mathbb{R}\to\mathbb{R}$ that is continuous in $A$ such that $\mathcal{A}(p)=\{c(a,p)\mid a\in A\}$ for all $p\in\mathbb{R}$. We can then replace $u\big(a(x-p)\big)$ by $u\big(c(a,p)(x-p)\big)$ and proceed as before, provided we change the assumption guaranteeing expected utilities to be finite and well-behaved to
\[\max_{a\in\mathcal{A}(p)}|u\big(c(a,p)(x-p)\big)|h(p,x,y)\leq d(p,x)\]
for all $(p,x)\in\mathbb{R}^2$ and all $y\in Y$ for a suitably chosen integrable function $d:\mathbb{R}\times\mathbb{R}\to \mathbb{R}$.

In general, there are problems with randomizing over measurable functions, as identified by \citet{MR0140636}.\footnote{A simpler proof of Aumann's result has been given by \citet{MR315076}. For a textbook treatment of Rao's proof, see \citet[Section 5.2]{MR3445285}.} The problem Aumann identifies is that the pointwise evaluation of measurable functions is, in general, not a jointly measurable mapping, no matter the $\sigma$-algebra one puts on the space of measurable functions. But limit orders are really equivalence classes of measurable functions that are not evaluated pointwise but by integration. This problem, therefore, does not affect our arguments.

\section{Proofs}\label{sec:proofs}

In what follows, we can assume without loss of generality that $\pi$ and $\xi$ are probability measures and not general $\sigma$-finite measures. The assumptions will still hold. To see this, start with some nontrivial $\sigma$-finite measures $\pi'$ and $\xi'$ and let $\pi$ and $\xi$ be probability measures such that $\pi'$ is absolutely continuous with respect to $\pi$ and such that $\xi$ is absolutely continuous with respect to $\xi'$. Let $r_\pi$ be a nonnegative Radon-Nikodym derivative of $\pi'$ with respect to $\pi$ and $r_\xi$ be a nonnegative Radon-Nikodym derivative of $\xi'$ with respect to $\xi$. Let $d':\mathbb{R}\times\mathbb{R}\to\mathbb{R}$ be given by $d'(p,x)=d(p,x)r_\pi(p)r_\xi(x)$ and let $h':\mathbb{R}\times\mathbb{R}\times Y\to\mathbb{R}$ be given by $h'(p,x,y)=h(p,x,y)r_\pi(p)r_\xi(x)$. Then, by Fubini's theorem,
\[\begin{split}\int\int d(p,x)~\mathrm d \pi'(p)~\mathrm d \xi'(x)&=\int d(p,x)~\mathrm d\pi'\otimes\xi'(p,x)\\
&= \int r_\pi(p)\int d(p,x) r_\xi(x)~\mathrm d\xi(x)~\mathrm d\pi(p)\\
&=\int r_\pi(p)r_\xi(p)d(p,x)~\mathrm d\pi\otimes\xi\\
&=\int d'~\mathrm d\pi\otimes\xi(p,x),\end{split}\]
so $d'$ is $\pi\otimes\xi$-integrable. Also, one can show by a similar argument that $\int h'(p,x,y)~\mathrm d \pi\otimes\xi(p,x)=\int h(p,x,y)~\mathrm d\lambda\otimes\lambda(p,x)$ for all $y\in Y$. Finally, by multiplying both sides the original uniform integrability inequality by $r_\pi(p)r_\xi(x)$, we obtain
\[|u(tp-tx)|h'(p,x,y)+|u(bx-bp)|h'(p,x,y)\leq d'(p,x)\]
for all $(p,x)\in\mathbb{R}^2$ and all $y\in Y$. So we can assume without loss of generality that our assumptions hold for the product of two probability measures.

\begin{proof}[Proof of Proposition \ref{wald}]
Let $\Delta_\pi(\mathbb{R}\times A)$ be the space of Borel probability measures on $\mathbb{R}\times A$ with $\mathbb{R}$-marginal $\pi$.
For $B\subseteq\mathbb{R}\times A$, let $1_B:\mathbb{R}\times A\to\{0,1\}$ be the corresponding indicator function. We define $\phi:\Delta(L)\to\Delta_\pi(\mathbb{R}\times A)$ by
\[\phi_\mu(B)=\int \int 1_B(p,l(p))~\mathrm d\pi(p)~\mathrm d\mu(l)\]
for each Borel set $B\subseteq\mathbb{R}\times A$.  It follows from \citet[Theorem 7.1]{balder1981mathematical} or the results in \citet[Section I]{MR658713} that $\phi$ is a surjection. Moreover,
 \[\begin{split}
V(\mu,y)&=\int\int u\big(l(p)(x-p)\big) h(p,x,y)\mathrm~ d\pi\otimes\xi(p,x)~\mathrm d\mu(l) \\
 &=\int\int\in tu\big(l(p)(x-p)\big) h(p,x,y)~\mathrm d\pi(p)~\mathrm d\mu(l)~\mathrm d\xi(x)\\
&=\int\int u\big(a(x-p)\big) h(p,x,y)~\mathrm d\phi_\mu(p,a)~\mathrm d\xi(x),
\end{split}\]
so we can study undominated mixed limit orders in terms of $\Delta_\pi(\mathbb{R}\times A)$. We can identify $\Delta_\pi(\mathbb{R}\times A)$ with a convex  and compact subset of a locally convex Hausdorff topological vector space as in \citet{MR942618} by endowing $\Delta_\pi(\mathbb{R}\times A)$ with the narrow topology on Young measures. It follows from the Scorza-Dragoni Theorem, see, for example, \citet*[Theorem 2.5.19]{MR2024162}, that this topology coincides with the usual topology of weak convergence of measures.

We can then, abusing notation a bit, treat $V$ as a continuous function $V:\Delta_\pi(\mathbb{R}\times A)\times Y\to\mathbb{R}$. We can also identify $Y$ homeomorphically with a closed subset of the weak*-compact set $\Delta(Y)$ via the embedding $y\mapsto\delta_y$. We then extend $V$ to a bilinear function $V^*:\Delta_\pi(\mathbb{R}\times A)\times\Delta(Y)\to\mathbb{R}$ via integration. The function $V^*$ is continuous by \citet[Theorem 2.5]{MR942618}. By \citet[Proposition 1.2]{MR1835574}, $\Delta(Y)$ is, under the embedding, the closed convex hull of $Y$.
So by \citet*[Lemma 1]{battigalli16n}, an element $\tau$ of $\Delta_\pi(\mathbb{R}\times A)$ is undominated if and only if $\tau\in\argmax V^*(\cdot,\beta)$ for some $\beta\in \Delta(Y)$.
\end{proof}

\begin{proof}[Proof of Proposition \ref{det}]
One direction is trivial in both cases. For the other direction, we first work with the assumption that $\pi$ is atomless. Assume that $\mu'$ strictly dominates $\mu$. By the Berge maximum theorem, the function $\kappa\mapsto\min_y V(\kappa,y)-V(\mu,y)$ is continuous. By assumption, it achieves a strictly positive value at $\mu'$. To finish the proof, we make use of the fact that the set of deterministic limit orders, embedded via the function $l\mapsto\phi_{\delta_l}$, is dense in $\Delta_\pi(\mathbb{R}\times A)$ when $\pi$ is nonatomic. This denseness is familiar from the optimal control literature, the classic reference being \citet[Theorem IV.2.6; 6]{MR0372708}. By this denseness, there exists a deterministic limit order $l$ such that $\min_y V(l,y)-V(\mu,y)>0$. Then $\mu$ is dominated by the deterministic limit order $l$.

Next, we assume that $u$ is weakly concave. It suffices to show that for each $\tau\in \Delta_\pi(\mathbb{R}\times A)$ there exists deterministic $\tau'\in \Delta_\pi(\mathbb{R}\times A)$ such that $V(\tau',y)\geq V(\tau,y)$ for all $y\in Y$. By the existence theorem for regular condition probabilities, \citet[Theorem 10.2.2.]{MR1932358}, $\tau$ corresponds to a modulo $\pi$-null sets unique measurable function $\kappa:\mathbb{R}\to  \Delta(A)$. Define $\kappa':\mathbb{R}\to\mathbb{R}$ by $\kappa'(p)=\int x~\mathrm d\kappa(p)$. We can treat $\kappa'$ as a degenerate (deterministic) regular condition probability that integrates back to some $\tau'\in \Delta_\pi(\mathbb{R}\times A)$. In the following, the inequality follows from Jensen's inequality; \citet[Theorem 10.2.6.]{MR1932358}. We have for all $y\in Y$
 \[\begin{split}
V(\tau,y)&=\int\int u\big(a(x-p)\big) h(p,x,y)~\mathrm d\tau(p,a)~\mathrm d\xi(x) \\
 &=\int\int\int u\big(a(x-p)\big) h(p,x,y)~\mathrm d\kappa_p(a)~\mathrm d\pi(p)~\mathrm d\xi(x) \\
&=\int\int\int u\big(a(x-p)\big)~\mathrm d\kappa_p(a) h(p,x,y)~\mathrm d\pi(p)~\mathrm d\xi(x) \\
&\leq\int\int u\big(\kappa'(p)(x-p)\big) h(p,x,y)~\mathrm d\pi(p)~\mathrm d\xi(x) \\
&=\int\int u\big(a(x-p)\big) h(p,x,y)~\mathrm d\tau'(p,a)~\mathrm d\xi(x).\\
&=V(\tau',y).
\end{split}\]
It follows that every mixed limit order dominated by $\tau$ must already be dominated by the deterministic limit order $\tau'$.
\end{proof}

\begin{proof}[Proof of Proposition \ref{beliefs}] Let $V^*:\Delta_\pi(\mathbb{R}\times A)\times\Delta(Y)\to\mathbb{R}$ be the bilinear continuous function introduced in the proof of Proposition \ref{wald}. An argument parallel to the proof of Proposition \ref{det} shows that there being no deterministic limit order $l$ such that $\int V(l',y)~\mathrm d\beta(y)>\int V(l,y)~\mathrm d\beta(y)$ for all $\beta\in\Pi$ implies that there is no $\mu\in\Delta_\pi(\mathbb{R}\times A)$ such that $V^*(\mu,\beta)>V^*(l,\beta)$ for all $\beta\in\Pi$. Applying \citet*[Lemma 1]{battigalli16n} to the restriction
$V^*:\Delta_\pi(\mathbb{R}\times A)\times\Pi\to\mathbb{R}$, we obtain a probability measure $\beta$ on $Y$ in the closed convex hull of $\Pi$ such that $l$ maximizes $V^*(\cdot,\beta)=\int V(\cdot,y)~\mathrm d\beta(y)$.
\end{proof}

\begin{proof}[Proof of Proposition \ref{simp}] As one sees from the proof of Proposition \ref{det}, it suffices to prove that the family of simple limit orders is dense in the space of  limit orders in the topology of convergence in measure. We can metrize the topology of convergence in measure by the \emph{Ky Fan metric} $\alpha$ given by $\alpha(l,l')=\inf\{\epsilon>0\mid \pi(|l-l'|>\epsilon)\leq\epsilon\}$; see \citet[Theorem 9.2.2.]{MR1932358}.
By a standard argument, one can approximate each limit order arbitrarily well by a simple function $f=\sum_{i=1}^m \lambda_i 1_{A_i}$ with the $A_i$ disjoint. Since each $A_i$ can be approximated from below by compact sets by Ulam's theorem, \citet[Theorem 7.1.4.]{MR1932358}, one can approximate $f$ arbitrarily well by a simple function $f'=\sum_{i=1}^m\lambda_i 1_{A_i'}$ with each $A_i'$ a compact subset of $A_i$. The function $f'$ is continuous on the compact set $\bigcup_{i=1}^m A_i'$ and has, by the Tietze extension theorem, \citet[Theorem 2.6.4.]{MR1932358}, a continuous extension to all of $\mathbb{R}$ whose range is contained in the convex hull of the range of $f'$. This way, one can approximate each limit order arbitrarily well by a continuous function. Clearly, one can approximate a continuous function arbitrarily well uniformly by a step function and thus a simple limit order on a compact set whose complement has arbitrarily small measure. The form of the Ky-Fan metric shows that this gives us the desired approximation.
\end{proof}

\section*{Appendix: Large State Spaces in the Motivating Example}\label{godhatesyou}

Our theory makes certain assumptions on the joint distributions of prices and final values. In particular, they must admit certain, well-behaved densities. This makes it impossible to model a setting in which an investor is minmaxed conditional on prices. The assumptions that rule out this possibility are made for certain arguments familiar from finite state-space models to still go through. But they are not essential to the conclusion. If we allow for richer models of the joint distribution of prices and final values, we can manually construct beliefs for a subjective expected utility maximizer that rationalize any behavior an ambiguity averse decision-maker might exhibit. But now, such behavior may include abstention for some prices.

Let $\pi$ be a Borel distribution on $\mathbb{R}$ that has full support and no mass points. We take $\pi$ to be the fixed distribution of prices. We then take the state space $Y$ to be $\Delta_\pi\big(\mathbb{R}\times\{-1,1\}\big)$, the space of Borel probability distributions on $\mathbb{R}\times\{-1,1\}$ with $\mathbb{R}$-marginal $\pi$. Limit orders are equivalence classes of measurable functions $l:\mathbb{R}\to\{-1,0,1\}$, equivalence meaning agreeing outside $\pi$-null sets. The limit order specifies the net amount being bought for each price. The expected payoff for a risk-neutral decision-maker choosing the limit order $l$ in state $y$ is
\[\int l(p)\cdot (x-p)~\mathrm d y(p,x).\]
Irrespective of the state, it is uniquely optimal to buy at prices smaller than $-1$ and to sell at prices larger than $1$. For every value between $-1$ and $1$, a minmax expected utility maximizer may strictly prefer to abstain. For every limit order $l$ that buys or sells for such prices, there exists a function $f$ from prices to $\{-1,1\}$ such that $f(p)=-l(p)$ for all $p$ such that $l(p)\neq 0$. For this function $f$, there is a unique state $y_f$, a probability distribution on $\mathbb{R}\times\{-1,1\}$ with $\mathbb{R}$-marginal $\pi$, that assigns probability one to the graph of $f$. In this state, the limit order $l$ does worse than the limit order that abstains from trading between $-1$ and $1$. So, with a sufficiently rich set of priors, a risk-neutral minmax expected utility maximizer will abstain. However, with such a rich state space, even risk-neutral subjective expected utility maximizers will also abstain from trading for prices between $-1$ and $1$ for some belief. To see this, we take the states in the support of the rationalizing belief to be the distributions corresponding to cutoff functions $f_z:\mathbb{R}\to\{-1,1\}$ given by $f_z(p)=1$ for $p \geq z$ and $f_z(p)=-1$ for $p < z$. Let $\mu$ be a distribution over such functions such that the corresponding cutoffs $z$ are uniformly chosen in $[-1,1]$ and let $F$ be the corresponding cumulative distribution function, given by $F(p)=(1+p)/2$ for $p$ between $-1$ and $1$. Then, the conditional expected final value of the asset at each price $p\in [-1,1]$ equals $p$, so that not trading at such prices is optimal for a risk-neutral subjective expected utility maximizer with belief $\mu$. To see this, note that the asset's conditional expected final value at a price $p\in (-1,1)$ is
\[\mu\big(f_z\mid z\leq p\big)\cdot 1\, +\, \mu\big(f_z\mid z> p\big)\cdot(-1)\]
\[=F(p)\cdot 1 + \big(1-F(p)\big)\cdot(-1)=p.\]

\bibliographystyle{plainnat}
\bibliography{References}

\end{document}